\documentclass[aps,pra,twocolumn,superscriptaddress,10pt]{revtex4-2}
\usepackage{silence}
\usepackage[utf8]{inputenc}
\usepackage{amsmath, amsthm, amsfonts, amssymb, bm, bbm}
\usepackage{hyperref, cleveref}
\usepackage{xcolor}
\usepackage{graphicx}  
\usepackage[english]{babel}
\usepackage{comment}
\usepackage{algorithm, algorithmic}
\usepackage{qcircuit}
\usepackage[bb=dsserif]{mathalpha}
\usepackage{array}
\newtheorem{theorem}{Theorem}

\newtheorem{lemma}{Lemma}

\usepackage{tikz}
\usetikzlibrary{arrows, arrows.meta}

\newcommand{\cgate}[1]{*+<.6em>{#1} \POS ="i","i"+UR;"i"+UL **\dir{-};"i"+DL **\dir{-};"i"+DR **\dir{-};"i"+UR **\dir{-},"i"}
\DeclareFontFamily{OT1}{pzc}{}
\DeclareFontShape{OT1}{pzc}{m}{it}{<-> s * [1.150] pzcmi7t}{}
\DeclareMathAlphabet{\mathpzc}{OT1}{pzc}{m}{it}

\DeclareMathOperator{\tr}{Tr}

\newcommand{\bb}{\mathbb}
\newcommand{\mc}{\mathcal}

\newcommand{\bra}[1]{\langle #1 |}  
\newcommand{\ket}[1]{| #1 \rangle}

\newcommand{\pr}{\ensuremath{\mathrm{Pr}}}

\makeatletter
\def\ketbra#1{\def\tempa{#1}\futurelet\next\ketbra@i}
\def\ketbra@i{\ifx\next\bgroup\expandafter\ketbra@ii\else\expandafter\ketbra@end\fi}
\def\ketbra@ii#1{\left| \tempa \middle\rangle\!\middle\langle #1 \right|}
\def\ketbra@end{\left| \tempa \middle\rangle\!\middle\langle \tempa \right|}
\makeatother

\makeatletter
\def\braket#1{\def\tempa{#1}\futurelet\next\braket@i}
\def\braket@i{\ifx\next\bgroup\expandafter\braket@ii\else\expandafter\braket@end\fi}
\def\braket@ii#1{\left\langle \tempa \middle| #1 \right\rangle}
\def\braket@end{\left\langle \tempa \middle| \tempa \right\rangle}
\makeatother

\newcommand{\dbra}[1]{\ensuremath{\left\langle\!\left\langle #1\right|\right.}}
\newcommand{\dket}[1]{\ensuremath{\left.\left|#1\right\rangle\!\right\rangle}}

\newcommand{\cketbra}[1]{\ensuremath{\left|#1\right)\!\left(#1\right|}}

\makeatletter
\def\dketbra#1{\def\tempa{#1}\futurelet\next\dketbra@i}
\def\dketbra@i{\ifx\next\bgroup\expandafter\dketbra@ii\else\expandafter\dketbra@end\fi}
\def\dketbra@ii#1{| \tempa \rangle\!\rangle\!\langle\!\langle #1 |}
\def\dketbra@end{| \tempa \rangle\!\rangle\!\langle\!\langle \tempa |}
\makeatother

\makeatletter
\def\dbraket#1{\def\tempa{#1}\futurelet\next\dbraket@i}
\def\dbraket@i{\ifx\next\bgroup\expandafter\dbraket@ii\else\expandafter\dbraket@end\fi}
\def\dbraket@ii#1{\langle\!\langle \tempa | #1 \rangle\!\rangle}
\def\dbraket@end{\langle\!\langle \tempa | \tempa \rangle\!\rangle}
\makeatother

\makeatletter
\def\cketbra#1{\def\tempa{#1}\futurelet\next\cketbra@i}
\def\cketbra@i{\ifx\next\bgroup\expandafter\cketbra@ii\else\expandafter\cketbra@end\fi}
\def\cketbra@ii#1{\left| \tempa \middle)\!( #1 \right|}
\def\cketbra@end{\left| \tempa \middle)\!\middle( \tempa \right|}
\makeatother

\makeatletter
\def\cbraket#1{\def\tempa{#1}\futurelet\next\cbraket@i}
\def\cbraket@i{\ifx\next\bgroup\expandafter\cbraket@ii\else\expandafter\cbraket@end\fi}
\def\cbraket@ii#1{\left( \tempa \middle| #1 \right)}
\def\cbraket@end{\left( \tempa \middle| \tempa \right)}
\makeatother

\begin{document}

\title{Benchmarking Quantum Instruments}
\author{Darian McLaren}
\affiliation{Institute for Quantum Computing, University of Waterloo, Waterloo, Ontario N2L 3G1, Canada}
\affiliation{Department of Applied Mathematics, University of Waterloo, Waterloo, Ontario N2L 3G1, Canada}
\author{Matthew A.\ Graydon}
\affiliation{Institute for Quantum Computing, University of Waterloo, Waterloo, Ontario N2L 3G1, Canada}
\affiliation{Department of Applied Mathematics, University of Waterloo, Waterloo, Ontario N2L 3G1, Canada}
\author{Ali Mahmoud}
\affiliation{National Research Council of Canada, Ottawa, Ontario  K1K 4P7, Canada}
\author{Joel J. Wallman}
\affiliation{Institute for Quantum Computing, University of Waterloo, Waterloo, Ontario N2L 3G1, Canada}
\affiliation{Department of Applied Mathematics, University of Waterloo, Waterloo, Ontario N2L 3G1, Canada}

\begin{abstract}
    Quantum measurements with feed-forward are crucial components of fault-tolerant quantum computers.
    We show how the error rate of such a measurement can be directly estimated by fitting the probability that successive randomly compiled measurements all return the ideal outcome.
    Unlike conventional randomized benchmarking experiments and alternative measurement characterization protocols, all the data can be obtained using a single sufficiently large number of successive measurements.
    We also prove that generalized Pauli fidelities are invariant under randomized compiling and can be combined with the error rate to characterize the underlying errors up to a gauge transformation that introduces an ambiguity between errors happening before or after measurements.
\end{abstract}

\maketitle

Quantum instruments (\textit{i.e}., nondestructive measurements with classical feed-forward~\cite{QI1}) are vital for detecting and correcting errors during quantum computation and communication. However, quantum instruments are also subject to noise, and so introduce errors into otherwise fault-tolerant schemes.
Therefore quantum instruments need to be characterized and any errors must be corrected to reliably process quantum information.

Historically, quantum characterization has primarily focused on quantum gates~\cite{eisert2020quantum}.
While complete characterizations are in principle possible using techniques such as process~\cite{Chuang1997} and gate-set tomography \cite{Merkel2013}, the complexity of quantum error processes renders such characterizations untenable for even moderately sized quantum systems.
These techniques can be generalized to characterize quantum instruments~\cite{QIQND,Rudinger2022,Pereira2022,pereira2023}, but still suffer the same fundamental scaling issue.

An alternative approach is to partially characterize errors by estimating figures of merit.
The canonical example of this approach is randomized benchmarking~\cite{magesan2011scalable}, wherein a unitary 2-design~\cite{Dankert2009} is used to ``twirl'' an arbitrary error channel into a global depolarizing channel with the same process fidelity.
The depolarizing channel is repeated a variable number of times and the process fidelity is estimated by fitting the survival probability to an exponential curve.
The primary drawback of this approach is that many entangling operations are needed to produce a global depolarizing channel, so that the error rate is only loosely connected to the performance of individual gates in a general circuit.

This drawback was overcome by the combination of cycle benchmarking~\cite{erhard2019characterizing} and randomized compiling~\cite{Wallman2016a}, which use local twirling operations to produce stochastic channels~\cite{mclaren2023stochastic}.
Unlike global depolarizing channels, stochastic error channels are described by a probability distribution and the action of the error depends on the input state.
When the twirling operations are Pauli operations (or Weyl operations for qudits), preparations and measurements in local Pauli eigenstates can be used to isolate portions of the state space in which the action of the error channel is described by a single parameter that can be amplified and estimated by fitting to an exponential decay.
These parameters are referred to as Pauli fidelities, and their average is exactly the process fidelity, which determines the error probability in randomly compiled circuits.
Moreover, additional processing of these Pauli fidelities can be performed to estimate the probabilities of specific errors~\cite{flammia2020}.

Randomized compiling and cycle benchmarking have both recently been generalized to quantum instruments~\cite{beale2023randomized,zhang2024,hines2024}.
These schemes are designed to characterize syndrome measurements and thus allow nontrivial circuits to entangle measured and unmeasured qubits.
Though useful, allowing entangling circuits and a subsystem of unmeasured qubits introduces significant complexity as: not all generalized Pauli fidelities are learnable~\cite{zhang2024};
the circuits depend on the Pauli fidelities to be learned; and benchmarking the unmeasured qubits increases the number of experiments.

In this paper, we present a simplified and more efficient protocol for characterizing instruments in the absence of other gates, and unmeasured qudits.
Specifically, we prove that the probability of obtaining the ideal outcome from the first $m$ of $k$ randomly compiled measurements decays exponentially in $k$, where, to first order, the base is the error rate.
We then show how additional post-processing can be used to learn more detailed properties of the noisy measurement.
As an example, we show that single-qubit instruments can be characterized up to a gauge ambiguity between errors occurring before or after the measurement.

The balance of this paper is structured as follows. In \cref{sec:prelims}, we define the multi-qudit Weyl operators and review how effective projectors can be introduced in a circuit. In \cref{sec:randomizedCompiling}, we prove that the generalized Pauli infidelities~\cite{zhang2024} are invariant under randomized compiling and completely characterize the randomly compiled instrument. In \cref{sec:diamondDistance}, we present our protocol for characterizing a noisy quantum instrument assuming no idling qudits along with a simulated example of its execution. We conclude with \cref{sec:discussion}. We also prove a bound on deviation from an ideal rank 1 projector in \cref{thm:ApproximateProjector} that may be of independent interest.

\section{PRELIMINARIES}
\label{sec:prelims}

We briefly set out notation and define the $n$-qudit Weyl operators and channels.
We will use the convention that operators acting on pure states will be denoted using capital Roman font (e.g., $X$, $Z$), whilst channels will be denoted using capital calligraphic letters (e.g., $\mc{M}$, $\mc{X}$, $\mc{Z}$).
For any operator $A$ acting on pure states, we use the same letter in calligraphic font to denote the corresponding channel that acts by conjugation, i.e., $\mc{A}(\rho) = A \rho A^\dagger$.

Let $d$ be a positive integer and $\oplus$ denote addition modulo $d$.
The characters of $\bb{Z}_d^n$ are the functions
\begin{align}
    \chi_z : \bb{Z}_d^n \to \bb{C} :: \chi_z(j) = \exp\left(\frac{2 \pi i z \cdot j}{d}\right),
\end{align}
for $z \in \bb{Z}_d^n$.
These functions satisfy Schur's orthogonality relations,
\begin{align}\label{eq:Schur}
    d^{-n} \sum_z \chi_j^*(z) \chi_k(z) = \delta_{j,k}.
\end{align}
Note that we will frequently use the identities 
\begin{align}
    \chi_a(b) = \chi^*_a(-b) = \chi^*_{-a}(b) = \chi_b(a).
\end{align}
The generalized single-qudit $X$ and $Z$ operators are
\begin{align}
    X= \sum_{j \in \bb{Z}_d} \ketbra{j\oplus 1}{j}, \quad Z = \sum_{j \in \bb{Z}_d} \chi_1(j) \ketbra{j}.
\end{align}
We use the vectorized power notation
\begin{align}
    A^a = \otimes_{i\in\bb{Z}_n} A^{a_i}
\end{align}
for an operator $A$ and a vector $a \in \bb{Z}_d^n$ to describe tensor products of powers of the $X$ and $Z$ operators, so that for $x, z \in \bb{Z}_d^n$,
\begin{align}\label{eq:ZToj}
    X^x = \sum_{j \in \bb{Z}_d^n} \ketbra{j\oplus x}{j}, \quad Z^z = \sum_{j \in \bb{Z}_d^n} \chi_z(j) \ketbra{j}.
\end{align}
We thus have the braiding relation
\begin{align}\label{eq:braiding}
    Z^z X^x = \chi_z(x) X^x Z^z
\end{align}
and, from \cref{eq:Schur}, the inversion formula
\begin{align}\label{eq:jToZ}
    \ketbra{j} = d^{-n} \sum_{z\in \bb{Z}_d^n} \chi_j^*(z) Z^z.
\end{align}
The $n$-qudit projective Weyl group is the set 
\begin{align}
    \{Z^z X^x: x, z \in \bb{Z}_d^n\}.
\end{align}

Quantum operations are completely positive trace non-increasing maps acting between the spaces of positive semi-definite operators. 
We will use the vectorization map 
\begin{align}
    \dket{A} = \sum_{a,b \in \bb{Z}_d^n} \bra{b} A \ket{a} \bigotimes_{j \in \bb{Z}_n} \ket{a_j b_j},
\end{align}
so that $\dket{A\otimes B} = \dket{A} \otimes \dket{B}$ and $\dbraket{A}{B} = \tr A^\dagger B$.
For clarity, we will use $\dket{k} = \dket{\ketbra{k}}$. 
If we wish to restrict to considering (normalized) quantum states we define a quantum channel to be any completely positive trace preserving map between spaces of density operators.
Of vital importance are the Weyl channels $\mc{X}^x$ and $\mc{Z}^z$, which satisfy
\begin{align}\label{eq:WeylChannels}
    \mc{X}^x &= d^{-n} \sum_{a, b} \chi^*_x(a)\dketbra{Z^a X^b}  \notag\\
    \mc{Z}^z &= d^{-n} \sum_{a, b} \chi_z(b) \dketbra{Z^a X^b}.
\end{align}
We can then construct effective projectors by inserting random Weyl operators, since by \cref{eq:Schur},
\begin{align}\label{eq:dephasing}
     \sum_x \chi_x(c) \mc{X}^x &=  \sum_b \dketbra{Z^c X^b} \notag\\
     \sum_z \mc{Z}^z &=  \sum_a \dketbra{Z^a}.
\end{align}
(Note that \cite[Lemma 1]{beale2023randomized} incorrectly differs from a normalization factor as the $d^{-n}$ normalization factor in \cref{eq:WeylChannels} was omitted.)

\section{Randomized compiling for quantum instruments}
\label{sec:randomizedCompiling}
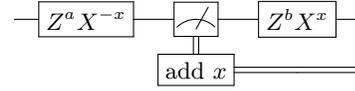
\begin{figure}[t!]
\begin{align*}
\Qcircuit @C=1em @R=.7em {
& \gate{Z^a X^{-x}} & \meter\cwx[1]  & \gate{Z^b X^x} & \qw  \\
& & \cgate{\mbox{add } x}  & \cw & \cw
}
\end{align*}
\caption{A randomly compiled $n$-qudit computational basis measurement where $a, b, x \in \bb{Z}_d^n$ are chosen uniformly at random. See \cite[Fig. 3]{beale2023randomized} for modifications for non-trivial gates on other unmeasured qudits.}
\label{fig:simplifiedMeasRC}
\end{figure}

We briefly review randomized compiling for quantum instruments as shown in \cref{fig:simplifiedMeasRC} and prove that the quantities learned in the protocols of \cite{zhang2024,hines2024} are invariant under randomized compiling.
For simplicity we will focus on the case where all qudits are measured.
A quantum instrument is an $n$-qudit quantum channel of the form
\begin{align}\label{eq:quantumInstrument}
    \mc{M} = \sum_{k \in \bb{Z}_d^n} \mc{M}_k \otimes \dket{k},
\end{align}
where $k$ represents the observed measurement outcome, and the maps $\mc{M}_k$ are trace non-increasing and sum to a trace-preserving channel. 
The generalized Pauli fidelities of an instrument~\cite{zhang2024} are defined to be
\begin{align}\label{eq:NuDef}
    \tilde{\nu}_{s, t}(\mc{M}) = d^{-n} \sum_{k} \chi^*_k(s-t) \mc{M}_{k, s, t}
\end{align}
where we define
\begin{align}\label{eq:MuKST}
    \mc{M}_{k, s, t} = \dbra{Z^t} \mc{M}_k \dket{Z^s}.
\end{align}
For an ideal instrument, $\mc{M}_k = \dketbra{k}$ and so for small errors we have $\mc{M}_{k,s,t} \approx \chi_k(s-t)$ and so $\tilde{\nu}_{s, t}(\mc{M}) \approx 1$ by \cref{eq:Schur}.
For a quantum instrument, we have
\begin{align}\label{eq:TracePreservingInstrument}
    \tilde{\nu}_{0,0}(\mc{M}) 
    = d^{-n} \tr (\sum_k \mc{M}_k)(I) 
    = 1.
\end{align}
As we now prove, the generalized Pauli fidelities of an instrument are invariant under randomized compiling with ideal single-qudit gates. Note that as is standard, we can accommodate gate-independent errors on the single-qubit gates by folding errors from the gates into the instrument, however, gate-dependent errors require a perturbative treatment.

\begin{theorem}
\label{thm:invariantGPFs}
    The generalized Pauli fidelities in \cref{eq:NuDef} are invariant under randomized compiling with ideal gates.
\end{theorem}

\begin{proof}
    From \cref{fig:simplifiedMeasRC}, our goal is to prove that $\tilde{\nu}_{s, t}\left([\mc{X}^x]^{\otimes 2} \mc{Z}^a \mc{M} \mc{Z}^b \mc{X}^{-x}\right)$ is independent of $a, b, x$.
    Due to the shifting of the observed outcome, we have
    \begin{align}\label{eq:XconjugatedInstrument}
    (\mc{X}^x)^{\otimes 2} &\mc{Z}^a \mc{M} \mc{Z}^b \mc{X}^{-x} \notag\\
    =& \sum_{k} \mc{X}^x\mc{Z}^a \mc{M}_{k} \mc{Z}^b \mc{X}^{-x} \otimes \dket{k+x} \notag\\
    =& \sum_{k} \mc{X}^x\mc{Z}^a \mc{M}_{k-x} \mc{Z}^b \mc{X}^{-x} \otimes \dket{k}.
    \end{align}
    Thus we have to shift the index $k$ to $k - x$ when substituting into \cref{eq:MuKST}, giving
    \begin{align}
    \tilde{\nu}_{s, t}&\left([\mc{X}^x]^{\otimes 2} \mc{Z}^a \mc{M} \mc{Z}^b \mc{X}^{-x}\right) \notag\\
    &= d^{-n} \sum_{k} \chi^*_k(s-t) \dbra{Z^t}\mc{X}^x \mc{Z}^a \mc{M}_{k-x} \mc{Z}^b \mc{X}^{-x} \dket{Z^s} \notag\\
    &= d^{-n} \sum_{k} \chi^*_k(s-t) \chi^*_{-x}(s - t)\dbra{Z^t} \mc{M}_{k-x} \dket{Z^s} \notag\\
    &= d^{-n} \sum_{k} \chi^*_k(s-t) \dbra{Z^t} \mc{M}_k \dket{Z^s},
    \end{align}
    which is independent of $a, b, x$ as required.
    In the above, we have used \cref{eq:WeylChannels} to obtain the second equality and then relabeled the terms in the sum to obtain the final equality.
\end{proof}

Having established that the generalized Pauli fidelities are invariant under randomized compiling, we now show that they completely determine the randomly compiled instrument, which is a subclass of uniform stochastic instruments~\cite{beale2023randomized} where all qudits are measured.

\begin{theorem}\label{thm:RCI}
    Randomized compiling with ideal gates maps an instrument $\mc{M}$ to
    \begin{align}\label{eq:RCZ}
        \hat{\mc{M}} = d^{-2n}\sum_{k, s, t} \chi_k(s-t) \tilde{\nu}_{s, t} \dketbra{Z^t}{Z^s} \otimes \dket{k}.
    \end{align}
    Equivalently,
     \begin{align}\label{eq:RCState}
        \hat{\mc{M}} = \sum_{k, a, b} \nu_{a, b} \dketbra{k+b}{k+a} \otimes \dket{k}
    \end{align}
    where
    \begin{align}\label{eq:ErrorRates}
        \nu_{a, b} = d^{-2n} \sum_{s, t} \tilde{\nu}_{s, t} \chi^*_s(a) \chi_t(b).        
    \end{align}
\end{theorem}

\begin{proof}
By \cref{eq:dephasing}, we have
\begin{align}
\mc{M}'_k
&= d^{-2n} \sum_{s, t}\mc{Z}^t \mc{M}_k \mc{Z}^s \notag\\
&= d^{-2n} \sum_{s, t}\mc{M}_{k, s, t} \dketbra{Z^t}{Z^s}.
\end{align}
Therefore averaging \cref{eq:XconjugatedInstrument} over $x, a, b \in \bb{Z}_d^n$ gives
\begin{align}\label{eq:RCInstrument}
    \hat{\mc{M}} &= d^{-3n} \sum_{a, b, x} (\mc{X}^x)^{\otimes 2} \mc{Z}^a \mc{M} \mc{Z}^b \mc{X}^{-x} \notag\\
    &= d^{-3n} \sum_{k, a, b, x} \mc{X}^x\mc{Z}^a \mc{M}_{k-x} \mc{Z}^b \mc{X}^{-x} \otimes \dket{k} \notag\\
    &= d^{-3n} \sum_{k, s, t, x} \mc{M}_{k-x, s, t} \mc{X}^x \dketbra{Z^t}{Z^s} \mc{X}^{-x} \otimes \dket{k} \notag\\
    &= d^{-2n} \sum_{k, s, t} \chi_k(s-t) \tilde{\nu}_{s, t} \dketbra{Z^t}{Z^s} \otimes \dket{k}
    \end{align}
as claimed.
To obtain the equivalent expression, we simply use \cref{eq:ZToj}.
\end{proof}

The $\nu_{a, b}$ are a probability distribution over $\bb{Z}_d^{2n}$~\cite{beale2023randomized}.
Note that $\dketbra{k+b}{k+a}$ corresponds to reporting the outcome $k$ when the system is actually in the state $k+a$, and then leaving the system in the state $k+b$, which we refer to as a register shift. 

Randomized compiling thus has two advantages.
First, it maps a noisy implementation of an ideal instrument into a simpler form with stochastic errors that are independent of the observed measurement outcome.
Secondly, and consequently, the diamond distance between a randomly compiled instrument $\hat{\mc{M}}$ and its ideal version $\mc{M}_{\rm id}$ is simply the probability $\epsilon$ of an error~\cite{mclaren2023stochastic}
\begin{align}\label{eq:diamondDistance}
    \frac{1}{2}\|\hat{\mc{M}} - \mc{M}_{\rm id}\|_\diamond = 1 - \nu_{0, 0} = \epsilon.
\end{align}

\section{Protocol for Characterizing Mid-circuit measurements}\label{sec:diamondDistance}
We now propose a protocol to characterize a noisy implementation of an ideal measurement.
Our protocol is based on the following subroutine for a fixed positive integer $m$ (the sequence length), closely resembling randomized benchmarking.
We describe the algorithm in the single-shot setting (i.e., choosing the gates independently for each experiment), although the generalization to the many-shot setting is straightforward.

\begin{algorithm}[H]
\begin{algorithmic}[1]
    \STATE Prepare the state $\ket{0}^{\otimes n}$.
    \STATE Set $\alpha_0 = 0$ 
    \STATE For each $i = 1, \ldots, m$
     \begin{enumerate}
        \item[a.] Choose $\alpha_i, \beta_i\in \bb{Z}_d^n$ uniformly at random.
        \item[b.] Apply the compiled operation $Z^{\beta_i} X^{\alpha_{i-1} - \alpha_i}$ to the measured qudits
        \item[c.] Measure the system in the computational basis and set $k_i = \alpha_i + o_i$ where $o_i$ is the observed outcome and $k_i$ is the de-randomized outcome. 
    \end{enumerate}
    \STATE Return the de-randomized outcomes $\vec{k}$.
\end{algorithmic}
\caption{Instrument benchmarking routine}\label{algo:mainProtocol}
\end{algorithm}

\begin{figure}[ht]
  \includegraphics[width=1\linewidth]{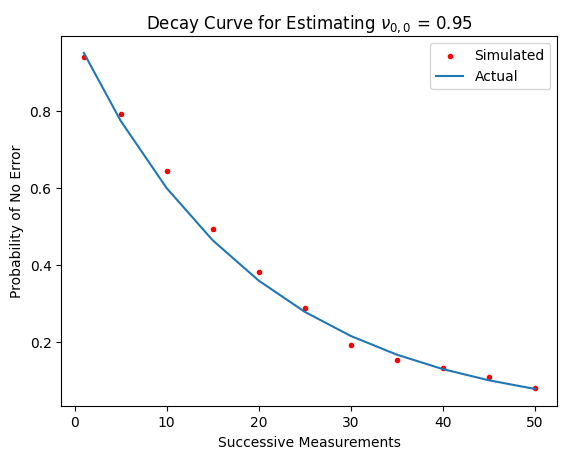}
  \caption{Decay curve for a simulated experiment of \cref{algo:mainProtocol} to estimate the theoretical value of $\nu_{0,0}=0.95$ when $\nu_{1,1}=0.05$ and all other entries of $N$ are zero.
  We numerically simulated \cref{algo:mainProtocol} with $m=50$ a total of 250 times and used the resulting measurement statistics to estimate the probability in \cref{thm:AllZeros}. The data points for smaller values of $m$ are obtained by marginalizing subsequent measurement outcomes. The resulting estimate was $\nu_{0,0}\approx 0.952$.}
  \label{fig:QPlot}
\end{figure}

\begin{lemma}\label{thm:probK}
Under ideal gates, the probability that \cref{algo:mainProtocol} returns $\vec{k}$ is
\begin{align*}
    \pr(\vec{k}) = \tr \hat{\mc{M}}_{\vec{k}}(\rho),
\end{align*}
where we define the ordered product
\begin{align*}
    \hat{\mc{M}}_{\vec{k}} = \hat{\mc{M}}_{k_m} \ldots \hat{\mc{M}}_{k_1}.
\end{align*}
\end{lemma}

\begin{proof}
    Let $\mc{M}$ be the noisy implementation of the measurement and $\hat{\mc{M}}$ be the randomly compiled version of $\mc{M}$.
    As we are assuming ideal gates, we can fold the ideal gates into the noisy instrument and average over them, so that the channel applied to the system when the outcomes $\vec{k}$ are observed is $\hat{\mc{M}}_{\vec{k}}$.
\end{proof}

From the expression in \cref{thm:probK} and \cref{thm:RCI}, we can perform post-processing in two ways to try and learn either the $\tilde{\nu}_{s, t}$ or the $\nu_{a, b}$.
\Cref{thm:RCI} gives two different expressions for $\hat{\mc{M}}_0$ corresponding to two different bases.

Taking $N$ to be the matrix such that $N_{b, a} = \nu_{a, b}$ will give a natural way of estimating  $\nu_{0,0}$, and thus the error probability $\epsilon$ by fitting an exponential curve to the probability of getting the all-zero vector.
We numerically illustrate this technique for a single-qubit system in \cref{fig:QPlot}. 

Unlike a standard randomized benchmarking experiment, one does not need to perform different experiments for each sequence length, rather, the results for different sequence lengths can be obtained from a single sequence length by marginalizing later measurements.
This optimization is not possible in the protocols outlined in from ~\cite{zhang2024,hines2024,hothem2024} as they require terminating measurements on the idling qudits.

\begin{theorem}\label{thm:AllZeros}
    Under ideal gates and provided $\epsilon < 1/3$, the probability that \cref{algo:mainProtocol} returns $\vec{0}$ is
\begin{align*}
    \pr(\vec{0}) = A\mu^m + O\left(\binom{m}{\xi}\epsilon^{m-\xi}\right),
\end{align*}    
where $\mu = 1 - \epsilon + O(\epsilon^2)$, $A$ is a constant, and $\xi + 1$ is the dimension of the largest Jordan block of $N$.
\end{theorem}

\begin{proof}
From \cref{thm:probK} and using the identity $\tr A^\dagger B = \dbraket{A}{B}$, we have
\begin{align}\label{eq:AllZeros}
    \pr(\vec{0}) = \dbra{I} \hat{\mc{M}}_0^m \dket{\rho}.
\end{align}
Let $\{e_a:a\in \bb{Z}_d^n\}$ be the standard orthonormal basis of $\bb{C}^{d^n}$ and
\begin{align}
    \vec{\rho} = \sum_a \dbraket{a}{\rho} e_a \notag\\
    \vec{I} = \sum_a \dbraket{a}{I} e_a.
\end{align}
Then \cref{eq:AllZeros} can be written as
\begin{align}
    \pr(\vec{0}) = \vec{I}^{\dagger} N^m \vec{\rho}.
\end{align}
By \cref{thm:ApproximateProjector}, $N$ has a unique maximal eigenvalue $\mu = \nu_{0,0} + O(\epsilon^2)$ with eigenvector $v$, while all other eigenvalues have modulus less than $\epsilon$.
Now let $SJS^{-1}$ be the Jordan decomposition of $N - \mu v v^\dagger$ and note that
\begin{align}
    N^m = \mu^m v v^\dagger + SJ^m S^{-1}.
\end{align}
To bound the elements of $J^m$, let $\ell$ be a fixed positive integer, $B_{j} \in \bb{C}^{\ell \times \ell}$ be the matrix that is zero except for $B_{j,i,i+j} = 1$ for $i = 0, \ldots, \ell - j$.
Noting that for any nonnegative integer $a$,
\begin{align}
    B_1^a = \begin{cases}
        B_a & a < \ell \\
        0 & \mbox{otherwise,}
    \end{cases}
\end{align}
we have
\begin{align}
    (\lambda I_\ell + B_1)^m = \sum_{j = 0}^{\ell-1} \binom{m}{j} \lambda^{m-j} B_j.
\end{align}
As all eigenvalues of $SJS^{-1}$ have modulus less than $\epsilon< 1/ 3$, all the entries of $J^m$ are at most $\binom{m}{\xi} \epsilon^{m-\xi}$ where $\xi + 1$ is the dimension of the largest Jordan block of $N - \mu v v^\dagger$, and hence of $N$.
\end{proof}

Fitting the probability of getting all zeros (after de-randomizing the outcome) gives a close approximation to $\nu_{0,0}$.
To learn the other values, we can use the technique from \cite{zhang2024,hines2024}, where we take the expectation value of the outcome-dependent phase in \cref{eq:Phase}.

\begin{theorem}\label{thm:LearnTildes}
    Under ideal gates, the expectation value of
    \begin{align}\label{eq:Phase}
        \prod_{j=1}^m \chi^*_{k_j}(c_j - c_{j+1})
    \end{align}
    for $c_1,\ldots, c_m \in \bb{Z}_d^n$ when $\vec{k}$ is the output of \cref{algo:mainProtocol} is
    \begin{align}\label{eq:ProdLambdaTilde}
    \kappa(\vec{c}) = \tr \left(Z^{-c_1} \rho\right) \prod_{j=1}^m \tilde{\nu}_{c_j, c_{j+1}},
    \end{align}
    where we set $c_{m+1} = 0 $.
\end{theorem}

\begin{proof}
    By \cref{thm:probK}, the expectation value of \cref{eq:Phase} under ideal gates is
    \begin{align}
        \kappa(\vec{c}) &= \sum_{\vec{k}} \dbra{I} \left[\prod_{j=m \to 1} \chi^*_{k_j}(c_j- c_{j+1}) \hat{\mc{M}}_{k_j}\right]\dket{\rho} \notag\\
        &= \dbra{I} \left[\prod_{j=m \to 1} \sum_{k_j} \chi^*_{k_j}(c_j-c_{j+1}) \hat{\mc{M}}_{k_j}\right]\dket{\rho}.
    \end{align}
    Substituting in \cref{eq:RCZ} and using \cref{eq:Schur}, we have
    \begin{align}
    \sum_{k_j} &\chi^*_{k_j}(c_j - c_{j+1}) \hat{\mc{M}}_{k_{j}} \notag\\
    &= d^{-2n} \sum_{k_j,s,t} \chi^*_{k_j}(c_j - c_{j+1})\chi_{k_j}(s-t) \tilde{\nu}_{s,t} \dketbra{Z^t}{Z^s} \notag\\
    &= d^{-n} \sum_{t_j} \tilde{\nu}_{t_j + c_j, t_j + c_{j+1}} \dketbra{Z^{t_j+c_{j+1}}}{Z^{t_j+c_j}}.
    \end{align}
    Noting that $I = Z^0$ and $\dbraket{Z^a}{Z^b} = d^n \delta_{a,b}$, we have $t_j = 0$ for all $j$ and thus
    \begin{align}
        \kappa(\vec{c}) &= \dbraket{Z^{c_1}}{\rho} \prod_{j=1}^m \tilde{\nu}_{c_j, c_{j+1}}
    \end{align}
    as claimed.
\end{proof}

We can use \cref{thm:LearnTildes} to learn all the $\tilde{\nu}_{a,b}$ for trace-preserving instruments up to a sign ambiguity as follows.
First recall that $\tilde{\nu}_{0, 0} = 1$, so we only need to learn the other values.
We choose $c_1 = 0$, so that $\tr (Z^{-c_1} \rho) = 1$.
We then note that taking the log of $\kappa(\vec{c})$ will give a sum of the generalized Pauli fidelities with multiplicities.
We can express this in vectorized form by defining the vectors $\vec{N}$ and $\vec{D}(\vec{c})$ where the $j =  a + d^n b$ entries of $\vec{N}$ and of $\vec{D}(\vec{c})$ are $\log \tilde{\nu}_{a, b}$ and the number of times that $\tilde{\nu}_{a, b}$ appears in \cref{eq:ProdLambdaTilde} respectively.
With these vectors, we then have
\begin{align}
    \log \kappa(\vec{c}) = \vec{D}(\vec{c}) \cdot \vec{N}.
\end{align}
Thus if we can choose $\vec{c}^j$ such that the matrix $D$ whose $j$th row is $\vec{D}(\vec{c}^j)$ is invertible, we can learn all the $\tilde{\nu}_{a, b}$ and thus, by \cref{eq:ErrorRates}, all the $\nu_{a, b}$.

Such a choice is generally impossible because certain errors will always appear in matched pairs, making a full characterization impossible.
Specifically, our circuits will not distinguish between a post-measurement register shift in one measurement and a pre-measurement register shift in a subsequent measurement.
Nevertheless, as we illustrate for a single qubit, we can exploit \cref{thm:AllZeros} to gain more information about the noisy instrument.
To fully characterize the noisy instrument, we need to learn $\tilde{\nu}_{a,b}$ for $a,b=0,1$.
As $\tilde{\nu}_{0,0} = 1$ for trace-preserving noise, we only need to learn the other three parameters.
We first choose $c_{2j} = 0$ and $c_{2j+1} = 1$ for $j = 0, \ldots, m-1$ to estimate
\begin{align}
    (\tilde{\nu}_{0,1}\tilde{\nu}_{1,0})^m,
\end{align}
and then fit to an exponential decay to learn $C = \tilde{\nu}_{0,1}\tilde{\nu}_{1,0}$.
We then choose $c_1 = 0$ and $c_j = 1$ for $j = 2, \ldots, m$ to estimate
\begin{align}
    \tilde{\nu}_{0,1}\tilde{\nu}_{1,0} \tilde{\nu}_{1,1}^{m-2},
\end{align}
and again fit to an exponential decay to learn $\tilde{\nu}_{1, 1}$.
We thus have 3 out of 4 required parameters and now use \cref{thm:AllZeros} to estimate $\nu_{0,0} = 1 - \epsilon$.
Then from \cref{eq:ErrorRates}, we have
\begin{align}
    4\nu_{0, 0} &= 1 + \tilde{\nu}_{1, 1} + \tilde{\nu}_{0,1} + \tilde{\nu}_{1,0}.
\end{align}
Rearranging and setting $B = 1 + \tilde{\nu}_{1,1} - 4 \nu_{0,0}$, we have
\begin{align}
    \tilde{\nu}^2_{0,1} + B\tilde{\nu}_{0,1} + C = 0
\end{align}
and so
\begin{align}\label{eq:signAmbiguity}
    \tilde{\nu}_{0,1} = \frac{-B \pm \sqrt{B^2 - 4 C}}{2}.
\end{align}
From the definition of $C$, we then have 
\begin{align}
    \tilde{\nu}_{1, 0} = \frac{-B \mp \sqrt{B^2 - 4 C}}{2},
\end{align}
although we cannot resolve the sign ambiguity.
From \cref{eq:ErrorRates}, the ambiguity between $\tilde{\nu}_{0,1}$ and $\tilde{\nu}_{1,0}$ propagates to an ambiguity between $\nu_{0,1}$ and $\nu_{1,0}$, that is, between whether it is more probable for an error to happen before or after the measurement.
However, the error rates $\nu_{0,0}$ and $\nu_{1,1}$ are unambiguous.
As we now show, this sign ambiguity is fundamental due to a gauge transformation.
If we map $\rho \to \mc{B}(\rho)$ and $\mc{M}_k \to \mc{B} \mc{M}_k \mc{B}^{-1}$ for all $k$, the probability that \cref{algo:mainProtocol} returns $\vec{k}$ is unchanged by \cref{thm:probK}.
Choosing
\begin{align}
    2\mc{B} = \dketbra{I} + \dketbra{X} + \dketbra{Y} + \frac{\tilde{\nu}_{1,0}}{\tilde{\nu}_{0,1}} \dketbra{Z}
\end{align}
swaps $\tilde{\nu}_{0,1}$ and $\tilde{\nu}_{1,0}$ and so no choice of $\vec{c}$ can distinguish between the two.
Thus the above procedure completely characterizes the underlying error rates up to a gauge transformation.

\section{Discussion}
\label{sec:discussion}
We presented a simple protocol to characterize noisy implementations of quantum measurements with feed-forward in the absence of idling qudits. In particular, we showed how to efficiently estimate the diamond distance between a randomly compiled quantum instrument and its ideal form via fitting the probability of no errors to a single exponential decay [\cref{thm:AllZeros}]. We also showed how the foregoing protocol could be supplemented with post-processing to learn (up to signs) the generalized Pauli fidelities of the noisy instrument in question [\cref{thm:LearnTildes}]. Indeed, we proved that the generalized Pauli fidelities are invariant under randomized compiling [\cref{thm:invariantGPFs}]. We derived, moreover, an expression for randomly compiled instruments given in terms of the characters of the Weyl group and the generalized Pauli fidelities [\cref{thm:RCI}]. Lastly, we derived an eigenvalue bound for nonnegative matrices that may be of independent interest [\cref{thm:ApproximateProjector}]. 

As we proved, our protocol provides an accurate estimate of the error rate.
The inexactness of the estimate is due to second-order errors where, for example, errors in the post-measurement state of one measurement cancel out classification errors in the next measurement.
These two error processes have equivalent effects on our protocol and are related by a gauge transform, resulting in the sign ambiguity in the generalized Pauli fidelities in \cref{eq:signAmbiguity}.
This gauge degree of freedom could likely be removed in a multi-qubit setting for sufficiently local errors by applying CNOT gates~\cite{PhysRevResearch.3.033285}, however, we leave this for future work.

While our protocol is simple, classical feed-forward is required to randomly compile the measurements. Randomly compiling the measurements is critical to achieving a meaningful characterization.
Without randomized compiling, the survival probability for $m$ measurements would still approximately decay exponentially at a rate close to the probability of no error conditioned on the initial state being the target state. 
However, this conditional probability is of limited usefulness for quantum error correction because a significant fraction of error syndromes should be observed.
Performing randomized compiling makes the error rate independent of the input state, which makes the error in, \textit{e.g.}, a syndrome measurement independent of the underlying syndrome and (approximately) equal to the error rate obtained using our protocol.
 
\section{Acknowledgments}
MG happily acknowledges helpful discussions with Stephen Vintskevich. This research was supported by the U.S. Army Research Office through grants W911NF-21-10007 and W911NF-20-S-0004, the Canada First Research Excellence Fund, the Government of Ontario, and the Government of Canada through NSERC.

\appendix

\section{Eigenvalue bound for nonnegative matrices}

\begin{figure}[b]
\begin{tikzpicture}
\draw[dashed] (0.8,0.6) node[below] {$A_{i,i}$} circle (0.6);
\draw[>=triangle 45, <->] (0.8,0.6) node[above] {$r_i(A)$} -- (1.4,0.6);
\draw[dashed] (0,0) circle (1.6);
\draw[dashed] (4.5,0) node[below] {$A_{0,0}$} circle (1.5);
\draw[>=triangle 45, <->] (0,0) -- node[below] {$R_i(A)$} (1.6,0);
\draw[>=triangle 45, <->] (4.5,0) -- node[right] {$r_0(A)$} (4.5,1.5);
\draw[->] (0,-2.5) -> (0,2.5);
\draw[->] (-2,0) -> (6.5,0);
\end{tikzpicture}
  \caption{The Gershgorin discs encircling $A_{i,i}$ and $A_{0,0}$ for a matrix $A$ with nonnegative elements dominated by $A_{0,0}$.
  The separation between the discs is at least $A_{0,0} - r_0(A) - R_i(A)$.}
  \label{fig:discSeparation}
\end{figure}
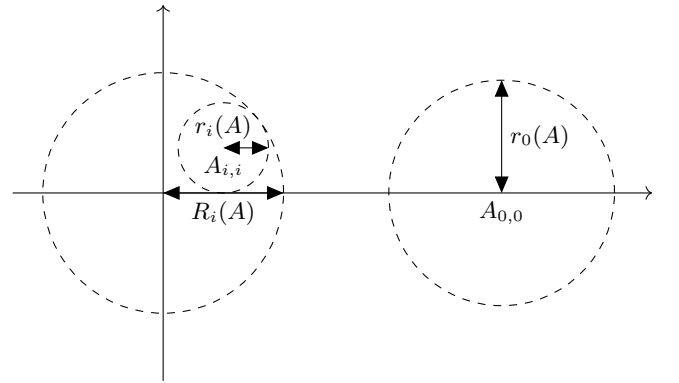

We prove an eigenvalue bound for nonnegative matrices that may be of independent interest.
To prove these bounds, we use Gershgorin discs and Brauer's eigenvalue bound.
The Gershgorin discs for a matrix $A$ are the sets
\begin{align}
    D_i(A) &= \{z: |z - A_{i, i}| \leq r_i(A)\}
\end{align}
where
\begin{align}
    r_i(A) = \sum_{j \neq i} |A_{i,j}|.
\end{align}
Informally, the Gershgorin disc theorem states that the eigenvalues of a matrix are located within discs centered on the diagonals.

\begin{theorem}[Gershgorin disc theorem \cite{Gershgorin31,horn2012matrix}]\label{thm:Gershgorin}
    Let $A \in \bb{C}^{\ell \times \ell}$ and $S \subseteq \bb{Z}_\ell$ be such that $\cup_{i \in S} D_i(A)$ and $\cup_{i\notin S} D_i(A)$ are disjoint.
    Then $\cup_{i \in S} D_i(A)$ contains $|S|$ eigenvalues of $A$ and $\cup_{i \notin S} D_i(A)$ contains $l-|S|$ eigenvalues of $A$.
\end{theorem}

While useful, the Gershgorin disc theorem is not sufficiently tight to obtain our desired bound on the dominant eigenvalue.
To strengthen the bound, we can use the ovals of Cassini
\begin{align}\label{eq:CasiniOvals}
    K_{i,j}(A) &= \{ z : |z - A_{i,i}| |z - A_{j,j}| \leq r_i(A) r_j(A) \}.
\end{align}

\begin{theorem}[Brauer's eigenvalue bound~\cite{Brauer47,horn2012matrix}]\label{thm:Brauer}
    The eigenvalues of a matrix $A \in \bb{C}^{\ell \times \ell}$ are contained within $\cup_{0 \leq i < j < \ell} K_{i,j}(A)$.
\end{theorem}

When one of the Gershgorin discs is well-separated from the other Gershgorin discs, we can use \cref{thm:Brauer} to tighten the Gershgorin disc theorem.
To formalize this, we quantify the separation between Gershgorin discs by defining
\begin{align}\label{eq:CasiniSeparation}
    k_{i,j}(A) = \min_{u \in D_i(A), v \in D_j(A)} | u - v |.
\end{align}

\begin{theorem}\label{thm:GershgorinBrauer}
    Let $A \in \bb{C}^{\ell \times \ell}$ be such that $r_i(A) < k_{0,i}(A)$ for all $i > 0$.
    Then there exists a unique eigenvalue $\mu$ of $A$ satisfying
    \begin{align*}
        |\mu - A_{0,0}| \leq \max_{i > 0} \frac{r_0(A) r_i(A)}{k_{0,i}(A)}.
    \end{align*}
\end{theorem}

\begin{proof}
As $r_i(A) < k_{0,i}(A)$ for all $i > 0$, $D_0(A)$ is disjoint from the other Gershgorin discs and so $D_0(A)$ contains one eigenvalue $\mu$ by \cref{thm:Gershgorin}.
That is,
\begin{align}
    |\mu - A_{0,0}| \leq r_0(A).
\end{align}
To complete the proof, we need to shrink the disc by a factor $r_i(A) / k_{0,i}(A)$, which is less than 1 by assumption.
By \cref{thm:Brauer}, there exists some $0 \leq i < j < \ell$ such that $\mu \in K_{i,j}(A)$, that is,
\begin{align}
    |\mu - A_{i,i}| \leq r_i(A) \frac{r_j(A)}{|\mu - A_{j,j}|} \leq \frac{r_i(A) r_j(A)}{k_{0,j}(A)}.
\end{align}
As $r_j(A) < k_{0,j}(A)$ by assumption, we have $\mu \in D_i(A)$ and so $i = 0$ as $\mu \in D_0(A)$ and $D_0(A)$ is disjoint from the other Gershgorin discs.
\end{proof}

We now specialize \cref{thm:GershgorinBrauer} to the case where the entries of $A$ form a probability distribution and $A_{0,0} \approx 1$.

\begin{theorem}\label{thm:ApproximateProjector}
    Let $\epsilon \in (0, 1/3)$ and $A \in \bb{C}^{\ell \times \ell}$ be a matrix with nonnegative entries such that $A_{0,0} = 1 - \epsilon$ and $\sum_{i,j} A_{i,j} = 1$.
    Then there is a unique maximal eigenvalue $\mu$ satisfying
    \begin{align*}
        |\mu - A_{0,0}| \leq \frac{\epsilon^2}{1 - 2\epsilon},
    \end{align*}
    while all other eigenvalues $\lambda$ satisfy $|\lambda| \leq \epsilon$. That is, $\mu = A_{0,0} + O(\epsilon^2)$.
\end{theorem}

\begin{proof}
For any matrix $A$ let $R_i(A)$ be the 1-norm of the $i$-th row of $A$, that is
\begin{align}
    R_i(A) = |A_{i, i}| + r_i(A).
\end{align}
Further, let $B(R_i(A))$ be the disk centered at the origin of radius $R_i(A)$. We thus trivially have $D_i(A) \subseteq B(R_i(A))$.
Since the separation between $D_0(A)$ and $B(R_i(A))$ is smaller than between $D_0(A)$ and $D_i(A)$, it follows that
\begin{align}
    k_{0,i}(A) \geq A_{0,0} - r_0(A) - R_i(A),
\end{align}
 where we note that $A_{0,0}>0$ by assumption. This is illustrated in \cref{fig:discSeparation}.
As $\sum_{i,j} A_{i,j} = 1$, we thus have $k_{0, i}(A) \geq 1 - 2\epsilon$ and $r_i(A) \leq \epsilon$, so that for $\epsilon \in (0, 1/3)$ we have $r_i(A) < k_{0, i}(A)$.
Thus, by \cref{thm:GershgorinBrauer},
\begin{align}
        |\mu - A_{0,0}| \leq \frac{\epsilon^2}{1 - 2\epsilon}.
\end{align}
By \cref{thm:Gershgorin}, all other eigenvalues $\lambda$ satisfy $|\lambda| \leq \epsilon$.
\end{proof}

\end{document}